\documentclass[a4paper,10pt]{article}%{amsart}
\usepackage[left=2.8cm,right=2.8cm,bottom=2.5cm,top=2.5cm]{geometry}
\usepackage[utf8]{inputenc}
\usepackage[english]{babel}
\usepackage[T1]{fontenc}

\usepackage{amssymb,amsmath,amsthm}
\usepackage{graphicx}
\usepackage{bm}
\usepackage{color}
\usepackage{authblk}

\newcommand{\de}{\partial}
\newcommand{\R}{\mathbb R}
\newcommand{\N}{\mathbb N}

\newcommand{\vc}[1]{\bm{#1}}

\newcommand{\leb}[2]{{L}^{#1}({#2})}
\newcommand{\sob}[2]{{W}^{#1}({#2})}
\newcommand{\crs}{\mathcal{A}}
\newcommand{\tpl}{\mathsf{w}}
\newcommand{\ha}{\mathcal{H}}
\newcommand{\loo}{{\Lambda[\tpl]}}
\newcommand{\wto}{\rightharpoonup}
\newcommand{\res}{\mathop{\hbox{\vrule height 7pt width .5pt depth 0pt
\vrule height .5pt width 6pt depth 0pt}}\nolimits}

\newtheorem{thm}{Theorem}[section]
\newtheorem{lem}[thm]{Lemma}

\theoremstyle{definition}
\newtheorem{defn}[thm]{Definition}

\numberwithin{equation}{section}

\begin{document}

\title{Solution of the Kirchhoff--Plateau problem}

%%%%For article
\author[1]{Giulio G.~Giusteri\thanks{\texttt{giulio.giusteri@oist.jp}}}
\author[2]{Luca Lussardi\thanks{\texttt{luca.lussardi@unicatt.it}}}
\author[1]{Eliot Fried\thanks{\texttt{eliot.fried@oist.jp}}}
\affil[1]{Mathematical Soft Matter Unit, Okinawa Institute of Science and Technology Graduate University, 1919-1 Tancha, Onna, Okinawa, 904-0495, Japan}
\affil[2]{Dipartimento di Matematica e Fisica, Universit\`a Cattolica del Sacro Cuore, via Musei 41, I-25121 Brescia, Italy}

%%%For AMSart
%\author{Eliot Fried, Giulio G.~Giusteri}
%\email{eliot.fried@oist.jp}
%\address[G.~G.~Giusteri]{Mathematical Soft Matter Unit, Okinawa Institute of Science and Technology Graduate University, 1919-1 Tancha, Onna, Okinawa, 904-0495, Japan}
%\email{giulio.giusteri@oist.jp}
%\author{Luca Lussardi}
%\address[L.~Lussardi]{Dipartimento di Matematica e Fisica, Universit\`a Cattolica del Sacro Cuore, via Musei 41, I-25121 Brescia, Italy}
%\email{luca.lussardi@unicatt.it}

\date{27 December 2016}

\maketitle

\begin{abstract}
The Kirchhoff--Plateau problem concerns the equilibrium shapes of a system in which a flexible filament in the form of a closed loop is spanned by a liquid film, with the filament being modeled as a Kirchhoff rod and the action of the spanning surface being solely due to surface tension. We establish the existence of an equilibrium shape that minimizes the total energy of the system under the physical constraint of non-interpenetration of matter, but allowing for points on the surface of the bounding loop to come into contact. In our treatment, the bounding loop retains a finite cross-sectional thickness and a nonvanishing volume, while the liquid film is represented by a set with finite two-dimensional Hausdorff measure. Moreover, the region where the liquid film touches the surface of the bounding loop is not prescribed a priori. Our mathematical results substantiate the physical relevance of the chosen model. Indeed, no matter how strong is the competition between surface tension and the elastic response of the filament, the system is always able to adjust to achieve a configuration that complies with the physical constraints encountered in experiments.
\end{abstract}

\section{Introduction}

Liquid films spanning rigid frames have been of longstanding interest to physicists and mathematicians, thanks to the sheer beauty of the countless observable shapes.  After the experimental investigations of Plateau \cite{Pla49}, anticipated by Lagrange's \cite{Lag60} theoretical treatment of the minimal surface problem, the first satisfactory proofs of the existence of a surface of least area bounded by a fixed contour were provided only in the twentieth century by Garnier \cite{Gar28}, Rad\'o \cite{Rad30}, and Douglas \cite{Dou31}. This formed a basis for a wealth of mathematical investigations regarding minimal surfaces, concerning various aspects and generalizations of the classical Plateau problem. The interested reader is referred to the treatises by Dierkes, Hildebrandt \& Sauvigny \cite{DieHilSau10} and Dierkes, Hildebrandt \& Tromba \cite{DieHilTro10} for a comprehensive review of the formative contributions.

{An important generalization to the situation in which the boundary of the minimal surface is not fixed but is constrained to lie on a prescribed manifold was initially treated by Courant \cite{Cou40} and Lewy \cite{Lew51}, whose work spurred a number of important mathematical contributions, as reviewed by Li \cite{Li15}. An existence theorem for a complementary generalization, in which part of the boundary is fixed and the remaining part is an inextensible but flexible string, was later proved by Alt \cite{Alt73}. In the present article, we introduce a problem which combines those generalizations. We consider situations in which the boundary of the minimal surface lies on a deformable manifold, namely the surface of an elastic loop. The filament forming the loop is assumed to be thin enough to be modeled faithfully by a Kirchhoff rod, an unshearable inextensible rod which can sustain bending of its midline and twisting of its cross-sections (see Antman \cite[Chapter VIII]{Antman2005}). This is a mathematically one-dimensional theory that describes a three-dimensional object, endowed with a nonvanishing volume, since the material cross-sections have nonvanishing area. The Kirchhoff--Plateau problem concerns the equilibrium shapes of a system composed by a closed Kirchhoff rod spanned by an area-minimizing surface.}

{In recent years, some attention has been drawn to the Kirchhoff--Plateau problem following a paper by Giomi \& Mahadevan \cite{GioMah12}, and the stability properties of flat circular solutions have been investigated, under various conditions regarding the material properties of the rod, by Chen \& Fried \cite{CheFri14}, Biria \& Fried \cite{BirFri14,BirFri15}, Giusteri, Franceschini \& Fried \cite{GiuFra16}, and Hoang \& Fried \cite{HoaFri16}.    
An existence result for a similar problem was given by Bernatzky \& Ye \cite{BerYe01}, but the elastic energy used therein fails to satisfy the basic physical requirement of invariance under superposed rigid transformations and thereby implicitly entails the appearance of unphysical forces. Moreover, a strong hypothesis is used to avoid the issue of self-contact. }

{Importantly, in all of these studies the boundary of the spanning surface is assumed to coincide with the rod midline and not to lie on the surface of the rod. This amounts to a slenderness assumption. Moreover, the surface is viewed as diffeomorphic to a disk, except by Bernatzky \& Ye \cite{BerYe01}, who employ the theory of currents. We relax both of these assumptions. Regarding the filament, we retain its three-dimensional nature for two physically well-justified reasons. First, there is a significant separation of scales between the typical thickness of the liquid film and the cross-sectional thickness of the filaments used in experimental investigations: a minimum of two orders of magnitude. Consequently, while the liquid film is still represented as two-dimensional, it seems appropriate to treat the bounding loop as a three-dimensional object. Second, the possibility of generating nontrivial shapes due to the interaction between the film and the bounding loop relies on the presence of anisotropic material properties of the filament, which are often associated with how its cross-sections are shaped.}

{Furthermore, the physical presence of the bounding loop requires a proper treatment of the constraint of non-interpenetration of matter, which is clearly at play in real experiments and even becomes essential, since the bounding loop can sustain large deflections but remains constrained when self-contact occurs. If, in particular, the relative strength of surface tension with respect to the elastic response of the filament becomes large, then the compliance of the mathematical solution with physical requirements can only be guaranteed by including the non-interpenetration constraint.} To include all these properties in a variational framework, we base our treatment of rod elasticity on Schuricht's \cite{Sch02} elegant approach. We introduce a minor simplification in the presentation of the constraint of local non-intrepenetration of matter, obtaining equivalent results specialized to the case of Kirchhoff rods.

In developing our variational approach, the most delicate point involves the treatment of the spanning surface. The physical phenomenon at play is the minimization of the liquid surface area due to the presence of a homogeneous surface tension. Various mathematical models have been proposed with different characteristics. Representing the surface via a mapping from a manifold into the ambient space, despite being the first successful approach, poses severe and completely unphysical limitations on the topology of the surface. To cope with this issue, the theories of integral currents and of varifolds were applied to the Plateau problem by Federer \& Fleming \cite{FedFle60} and Almgren \cite{Alm68}, respectively. However, their approaches also fail either to cover all the physical soap film solutions to the Plateau problem or to furnish a sufficiently general existence result. An alternate route was initiated by Reifenberg \cite{Rei60}, who treated the surface as a point set that minimizes the two-dimensional Hausdorff measure. This purely spatial point of view, adopted also by De Pauw \cite{DeP09} and David \cite{Dav14}, deals nicely with the topology of solutions, but makes it difficult to handle a generic boundary condition. A more complete treatment, which covers all the soap film solutions to the Plateau problem, has been developed by Harrison \cite{Har14} and Harrison \& Pugh \cite{HarPug13} using differential chains. In that setting, it is possible to consider generic bounding curves due to the presence of a well-defined boundary operator.

In relation to the problem at hand, the approach to the Plateau problem due to De Lellis, Ghiraldin \& Maggi \cite{DeLGhi14} proves to be superior in many respects. First and foremost, De Lellis, Ghiraldin \& Maggi \cite{DeLGhi14} treat the surface as the support of a Radon measure, adopting a spatial point of view, and thereby obtain the optimal soap film regularity defined by Almgren \cite{Alm76} and Taylor \cite{Tay76}. Moreover, their definition of the spanning conditions, built on ideas of Harrison \cite{Har14} (further developed by Harrison \& Pugh \cite{HarPug13,HarPug16}), allows for an apt treatment of the free-boundary problem, whereas all prior approaches become rather difficult to use when the boundary of the spanning surface is not prescribed. Finally, their strategy has the physically relevant property of being insensitive to changes in the topology of the spanning surface, which can easily be observed during the relaxation to equilibrium of the system with an elastic bounding loop. Another interesting approach that allows for treatment of the free-boundary problem is provided in a recent paper by Fang~\cite{Fan16}, but its generalization to the case of a deformable bounding loop seems to require a more sophisticated apparatus than that employed in our treatment.

{We present the formulation of the Kirchhoff--Plateau problem in Section \ref{sec:formulation}, where we introduce the energy functionals pertaining to the parts of our system and a suitable expression of the spanning condition.
Of particular importance is also the introduction of various physical and topological constraints regarding the non-interpenetration of matter and the link and knot type of the bounding loop. Indeed, those natural requirements necessitate a variational approach to the problem, since they entail an inherent lack of smoothness, leaving open the question of the validity of the Euler--Lagrange equations at energy-minimizing configurations. Furthermore, they hinder the convexity properties of a seemingly innocuous functional. This becomes even more evident when the surface energy is added to the picture, so that the presence of multiple stable and unstable equilibria cannot in general be precluded.}

{We prove our main existence result in Section \ref{sec:existence}, subsequent to establishing some preliminary facts.
The essential feature of our treatment is a dimensional reduction which is performed by expressing the total energy of the system (bounding loop plus spanning surface) as a functional of the geometric descriptors of the bounding loop only. This is done by introducing a strongly nonlocal term entailing the minimization of the surface energy for a fixed shape of the bounding loop. Clearly, this step is justified by the existence of a solution to the problem of finding an area-minimizing surface spanning a three-dimensional bounding loop. We prove this (Theorem \ref{thm:Plateau}) as a direct application of a result established by De Lellis, Ghiraldin \& Maggi \cite[Theorem 4]{DeLGhi14}.
Subsequently, it is necessary to adapt the arguments of De Lellis, Ghiraldin \& Maggi \cite{DeLGhi14} to our situation, in which the set that is spanned by the surface changes along minimizing sequences for the Kirchhoff--Plateau energy. This is accomplished in Lemma \ref{lem:diagonal} and Lemma \ref{lem:intersection} and, based on these results, we establish the lower semicontinuity property that is needed to establish, in Theorem \ref{thm:existence}, the existence of a solution to the Kirchhoff--Plateau problem.}

%%%%%%%  SECTION 2  %%%%%%%%%%%%%%%%%%%%%%%%%%%%%%%
\section{Formulation of the problem}\label{sec:formulation}

We seek to study the existence of a stable configuration of a liquid film that spans a flexible loop. Such a configuration is in reality metastable, since the liquid film will eventually break after becoming sufficiently thin, but we confine our attention to what happens before any such catastrophic event.

The flexibility of the bounding loop represents a major difference between our problem and the Plateau problem, in which the liquid film spans a fixed boundary. It is therefore essential to model the elastic behavior of the loop in response to deformations, requiring a physical description much more sophisticated than that sufficient for a fixed boundary. 
In particular, we consider a loop formed by a slender filament with a nonvanishing cross-sectional thickness and subject to the physical constraint of non-interpenetration of matter. As discussed below, it is still reasonable to approximate the liquid film by a surface, since its thickness is typically much smaller than the cross-sectional thickness of the bounding loop, but an appropriate definition of how the surface \emph{spans the bounding loop} is necessary.

We next introduce the precise mathematical definitions needed to formulate the Kirchhoff--Plateau problem in a way that takes into account the physical requirements mentioned above. In so doing, we follow a variational approach by defining the energy of the different components of our system.

\subsection{Energy of the bounding loop}

\subsubsection{Preliminary considerations and assumptions}

The main assumption we impose on the bounding loop is that its length is much larger than the characteristic thickness of its cross-sections, which allows us to employ the theory of rods in our description. Within that theory, as presented for example by Antman~\cite{Antman2005},
a rod is fully described by a curve in the three-dimensional Euclidean space, called the {midline}, a family of two-dimensional sets, describing the {material cross-section} at each point of the midline, and a family of {material frames}, encoding how the cross-sections are ``appended'' to the midline. Such a family of material frames corresponds also to a curve in the group of rotations of the three-dimensional space.

We also assume that the rod is unshearable (namely, that the material cross-section at any point of the midline lies in the plane orthogonal to the midline at that point) and that its midline is inextensible. Together, these assumptions amount to choosing a Kirchhoff rod as a model for the filament from which the bounding loop is made. Under these assumptions, the shape of the loop is uniquely determined by assigning the shape of the cross-sections and three scalar fields: two {flexural densities} $\kappa_1$ and $\kappa_2$ together with a {twist density} $\omega$. 
From these fields we can reconstruct the midline $\vc x$ and a director field $\vc d$, orthogonal to the tangent field $\vc t:=\vc x'$, that gives the material frame as $\{(\vc t(s),\vc d(s),\vc t(s)\times\vc d(s)):s\in[0,L]\}$, where $s$ is the arc-length parameter and $L$ is the total length of the midline. Indeed, once suitable conditions at $s=0$ are assumed, the fields $\vc x$ and $\vc d$ are the unique solution of the system of ordinary differential equations
\begin{equation}\label{eq:Cauchy}
\left\{\begin{aligned}
%&\vc x(0)=\vc x_0\,,\\
%&\vc t(0)=\vc t_0\,,\\
%&\vc d(0)=\vc d_0\,,\\
&\vc x'(s)=\vc t(s)\,,\\
&\vc t'(s)=\kappa_1(s)\vc d(s)+\kappa_2(s)\vc t(s)\times\vc d(s)\,,\\
&\vc d'(s)=\omega(s)\vc t(s)\times\vc d(s)-\kappa_1(s)\vc t(s)\,,
\end{aligned}\right.
\end{equation}
for $s$ in $[0,L]$, supplemented by the initial conditions
\begin{equation}\label{eq:CauchyIC}
\left\{\begin{aligned}
&\vc x(0)=\vc x_0\,,\\
&\vc t(0)=\vc t_0\,,\\
&\vc d(0)=\vc d_0\,.
\end{aligned}\right.
\end{equation}
It is not a priori granted that the solution to \eqref{eq:Cauchy}--\eqref{eq:CauchyIC} describes a closed loop. This property, being essential to the treatment of the Kirchhoff--Plateau problem, is imposed later as a constraint on the variational problem.

For the variational problem that we plan to investigate, it is convenient to assume that each of the densities $\kappa_1$, $\kappa_2$, and $\omega$ belongs to the Lebesgue space $\leb{p}{[0,L];\R}$ for some $p$ in $(1,\infty)$. This, by a classical result of Carath\'eodory (see, for instance, Hartman \cite{Hartman1982}), ensures that \eqref{eq:Cauchy}--\eqref{eq:CauchyIC} has a unique solution, with $\vc x$ in $\sob{2,p}{[0,L];\R^3}$ and $\vc d$ in $\sob{1,p}{[0,L];\R^3}$, where $\sob{n,p}{[0,L];\R^d}$ denotes the Sobolev space of measurable functions from $[0,L]$ to $\R^d$ with $n$ distributional derivatives in $\leb{p}{[0,L];\R^d}$. We further assume $|\vc t_0|=|\vc d_0|=1$. On this basis, we can use the structure of \eqref{eq:Cauchy}$_{2,3}$ to prove that $|\vc t(s)|=|\vc d(s)|=1$  for every $s$ in $[0,L]$.

The material cross-section at each $s$ is given by a compact simply connected domain $\crs(s)$ of $\R^2$ such that the origin $\vc 0_2$ of $\R^2$ belongs to $\mathrm{int}(\crs(s))$. A rod of finite cross-sectional thickness can then be described as the image in the three-dimensional Euclidean space of the set
\begin{equation*}
\Omega:=\big\{(s,\zeta_1,\zeta_2):s\in[0,L]\text{ and }(\zeta_1,\zeta_2)\in\crs(s)\big\}
\end{equation*}
through the map
\begin{equation}\label{eq:configuration}
\vc p(s,\zeta_1,\zeta_2):=\vc x(s)+\zeta_1\vc d(s)+\zeta_2\vc t(s)\times\vc d(s)\,.
\end{equation}
By our assumptions, there exists an $R>0$ such that $|\zeta_1|<R$ and $|\zeta_2|<R$ for any $(s,\zeta_1,\zeta_2)$ in $\Omega$. We remark that, for the rod model to be a faithful representation of the mechanics of the filament from which the bounding loop is made, it is necessary that the maximum thickness $R$ be small compared to the length $L$ of the loop.

Once the family of material cross-sections is assigned, any configuration of the rod in the (three-dimensional) ambient space corresponds to an element
\[
\tpl:=((\kappa_1,\kappa_2,\omega),\vc x_0,\vc t_0,\vc d_0)
\]
belonging to the Banach space
\[
V:=\leb{p}{[0,L];\R^3}\times\R^3\times\R^3\times\R^3\,.
\]
Whereas all of the information regarding the \emph{shape} of the rod is encoded in the flexural and twist densities, namely in the component $\tpl_1=(\kappa_1,\kappa_2,\omega)$ of $\tpl$ belonging to $\leb{p}{[0,L];\R^3}$, the {clamping parameters} $\vc x_0$, $\vc t_0$, and $\vc d_0$ determine how the rod is translated and rotated in the ambient space. We denote the midline, the director field, and the rod configuration computed from \eqref{eq:Cauchy}--\eqref{eq:CauchyIC} and \eqref{eq:configuration} for a given $\tpl$ in $V$ as $\vc x[\tpl]$, $\vc d[\tpl]$, and $\vc p[\tpl]$, respectively, and the bounding loop occupies the subset $\loo:=\vc p[\tpl](\Omega)$ of $\R^3$.

\subsubsection{Individual contributions to the energy of the bounding loop}

We are now positioned to define the energy of the rod as a functional on $V$. We consider three contributions entering this functional in an additive way: (i) the stored elastic energy, related with shape modifications; (ii) the non-interpenetration constraint; and (iii) the potential energy of an external load, such as the weight of the rod.

The first term, being related only to shape deformations, can be expressed as the integral of an elastic energy density which depends only on $s$ and $\tpl_1$. We then introduce $f:\R^3\times[0,L]\to\R\cup\{+\infty\}$ and define the shape energy of the bounding loop as
\[
E_\mathrm{sh}(\tpl):=\int_0^Lf(\tpl_1(s),s)\,ds\,,
\]
and we assume that $f(\cdot,s)$ is continuous and convex for any $s$ in $[0,L]$, that $f(\mathsf{a},\cdot)$ is measurable for any $\mathsf{a}$ in $\R^3$, and that $f(\mathsf{a},s)$ is uniformly bounded below by a constant. These assumptions guarantee that $E_\mathrm{sh}$ is weakly lower semicontinuous on the reflexive Banach space $V$---a key property in applying the direct method of the calculus of variations. To ensure the necessary coercivity, we also impose the natural growth condition
\[
f(\mathsf{a},s)\geq C_1|\mathsf{a}|^p+C_2\,,
\] 
with $C_1>0$ and $C_2$ in $\R$.

To include the local non-interpenetration constraint, we use the characteristic function of the closed subset $N$ of $V$ containing those elements $\tpl$ such that
\begin{equation}\label{eq:ni}
\max_{(\zeta_1,\zeta_2)\in\crs(s)}\big(\zeta_1\kappa_2(s)-\zeta_2\kappa_1(s)\big) \leq 1
\end{equation}
for almost every $s$ in $[0,L]$. We thus add to the energy of the loop the term
\[
E_{\mathrm{ni}}(\tpl):=
\begin{cases}
0 & \text{if }\tpl\in N\,,\\
+\infty & \text{if }\tpl\in V\setminus N\,.
\end{cases}
\]
Condition \eqref{eq:ni} can be derived by relaxing the standard requirements that $\vc p[\tpl]$ be orientation preserving and locally injective (see Antman~\cite[Chap.~VIII.7]{Antman2005} and Schuricht~\cite{Sch02}). We will prove, in Theorem~\ref{thm:injectivity}, that our penalization strategy (which differs slightly from that of Schuricht~\cite{Sch02}) is sufficient to guarantee the local injectivity of $\vc p[\tpl]$ on $\mathrm{int}(\Omega)$ for configurations with finite energy. The concept of local interpenetration of matter is illustrated in Figure \ref{fig:injectivity}.

\begin{figure}
\begin{center}
\includegraphics[width=0.7\textwidth]{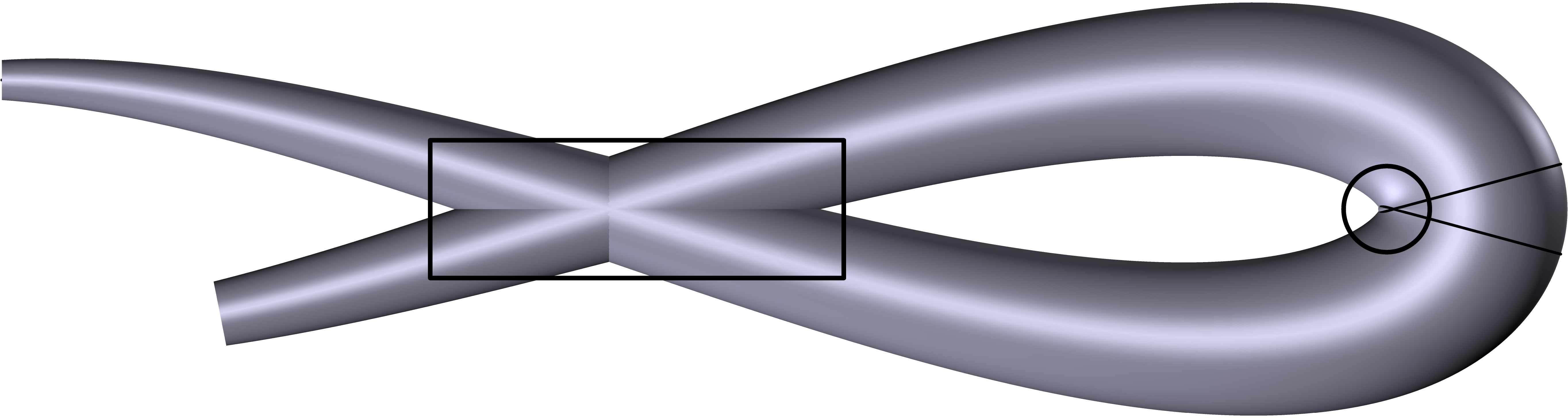}\hspace{.2cm}
\includegraphics[width=0.18\textwidth]{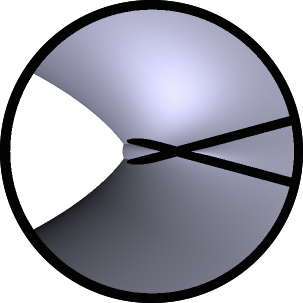}
\end{center}
\caption{Local injectivity fails when adjacent material cross-sections interpenetrate. This is illustrated in the encircled region (magnified on the right) where two material cross-sections, traced on the surface, overlap, due to the excessive curvature of the midline. Global injectivity fails also when interpenetration occurs between cross-sections that lie far apart along the midline. This situation, in which local injectivity is not hindered, is depicted in the framed region.}
\label{fig:injectivity}
\end{figure}

Finally, we account for
the effects of the weight of the rod by considering the potential energy
\[
E_\mathrm{g}(\tpl):=-\int_\Omega \rho(s,\zeta_1,\zeta_2)\vc g\cdot\vc p[\tpl](s,\zeta_1,\zeta_2)
 \, d(s,\zeta_1,\zeta_2)\,,
\]
where $\rho>0$ represents the mass density at each point of the rod and the vector $\vc g$ represents the constant gravitational acceleration.

From the above definitions, we obtain a weakly lower semicontinuous functional
\begin{equation}\label{eq:Eloop}
E_\mathrm{loop}(\tpl):=E_\mathrm{sh}(\tpl)+E_{\mathrm{ni}}(\tpl)+E_\mathrm{g}(\tpl)
\end{equation}
representing the total energy of the bounding loop.

\subsection{Energy of the liquid film}

A fundamental tenet in physical chemistry is that he process of building an interface between two immiscible substances is accompanied by an energetic cost. Roughly speaking, each portion of each given substance prefers to be surrounded by the same substance and some interfacial energy is developed whenever this is not the case. For this reason, a droplet of water surrounded by air tends to assume a spherical shape: for a fixed droplet volume, it minimizes the area of the interface.

To produce a liquid film, it is necessary to counteract this tendency, since doing so entails stretching the droplet, necessarily increasing the area of the liquid/air interface. We operate in two ways. On one hand, the solid bounding loop provides a third substance to which the liquid is attracted, since the energy density per unit area of the liquid/solid interface is lower than that associated with the liquid/air interface. On the other hand, it is possible to tamper a bit with the liquid to further reduce the energetic cost of the liquid/air interface.

The second objective is accomplished by adding a small amount of surfactant to the liquid. Since the energy density per unit area of the liquid/air interface is lower for a higher surfactant concentration, surfactant molecules migrate towards the interface, leaving water in the bulk. In the liquid film configuration this produces two leaflets of surfactant phase (that lower the interfacial energy) covering a thin water layer (that provides a significant cohesion to the structure).
In this context, we can define the {surface tension} $\sigma$ of the liquid as the energy density per unit area of the liquid/air interface. It is physically reasonable to assume that $\sigma$ is a homogeneous positive quantity, representing the ratio between the total interfacial energy and the surface area of the interface.

To arrive at a mathematical model for the liquid film, we now combine a geometric approximation with the notion of interfacial energy discussed above. We assume that the thickness of the film (two surfactant leaflets plus the water layer) is negligible and we represent it as a two-dimensional object $S$, but we keep track of the fact that it is built with two surfactant leaflets. We thus define the energy of the liquid film as
\begin{equation}\label{eq:Efilm}
E_\mathrm{film}(S):=2\sigma \ha^2(S)\,,
\end{equation}
where $\ha^d$ represents the $d$-dimensional Hausdorff measure.

Two remarks are in order. First, there is no evidence in \eqref{eq:Efilm} of the energy associated with the liquid/solid interface along which the film is in contact with the bounding loop. This is, in principle, a significant contribution, but the energy barrier that must be overcome to detach the film from the bounding loop is so high that its effect can be replaced by the spanning condition, discussed in the next section, designed to prevent detachment, encoding in essence the infinite height of the aforementioned barrier. The small corrections to that energy due to the size and shape of the liquid/solid interface are negligible as far as detachment is concerned. The influence of those corrections on the shape energy of the film is also negligible if the film is assumed to be of vanishing thickness. Indeed, they would influence the shape of the liquid/solid interface at length scales that are smaller than its typical thickness and, hence, not captured by our model.

\begin{figure}
\begin{center}
\setlength{\unitlength}{1cm}
\begin{picture}(10,5)
\put(0,0){\includegraphics[width=10cm]{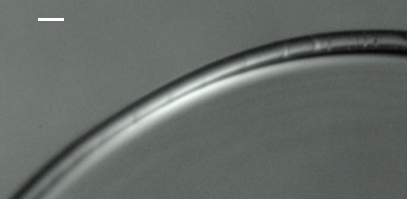}}
\thicklines
{\color{white}\put(6,0.7){\vector(-2,3){1.}}
\put(0.56,3.8){{\large0.2~mm}}}
\end{picture}
\end{center}
\caption{The thickness of the liquid film is orders of magnitude smaller than the cross-sectional thickness of the bounding loop. The white arrow points to the bright region where the thickness of the soap film slightly increases just before touching the bounding loop, realized using commercial fishing line with cross-sectional diameter of approximately $0.2$ mm. Since most of the filament surface is not covered by the liquid film, the thickness of the latter must be considerably less than $0.2$ mm.}
\label{fig:thicknesscomparison}
\end{figure}
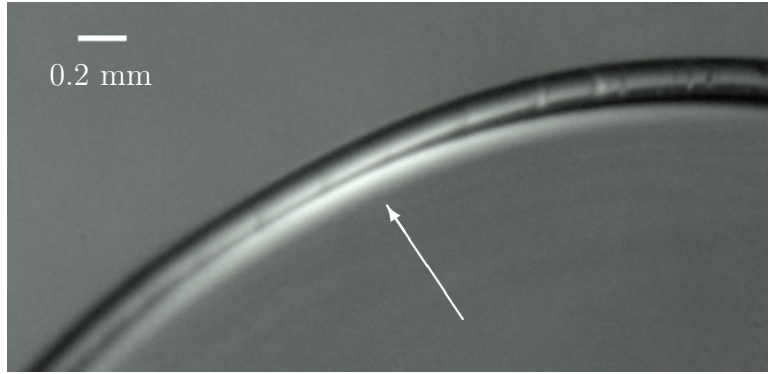

Second, it is important to justify the vanishing-thickness approximation, especially in view of the nonvanishing thickness that we attribute to the rod modeling the bounding loop. The typical thickness of a soap film is on the sub-micron scale, while a bounding loop made of a strand of human hair would have a cross-section with characteristic thickness of some tens of microns: this indicates that in many practical examples of liquid films bounded by flexible loops the characteristic thickness of the loop is at least two orders of magnitude greater than the thickness of the film (see Figure~\ref{fig:thicknesscomparison}). 

\subsection{The spanning condition}\label{sec:spanning}

We now provide a precise mathematical formulation of the conditions stipulating that the liquid film spans the bounding loop. In so doing, we borrow an elegant idea, that exploits notions of algebraic topology, introduced by Harrison \cite{Har14} (and further developed by Harrison \& Pugh \cite{HarPug13,HarPug16}).
For our application it is convenient to present that idea in the form provided by De Lellis, Ghiraldin \& Maggi \cite{DeLGhi14}, specialized to the particular setting in which the ambient space is three-dimensional.

\begin{defn}
Let $H$ be a closed set in $\R^3$ and consider the family
\[
\mathcal C_H:=\big\{\gamma:S^1\to\R^3\setminus H:\gamma\text{ is a smooth embedding of }S^1\text{ into }\R^3\big\}\,.
\]

Then, any subset $\mathcal C$ of $\mathcal C_H$ is closed by homotopy (with respect to $H$) if $\mathcal C$ contains all the elements $\hat\gamma$ of $\mathcal C_H$ belonging to the same homotopy class $[\gamma]$ of any $\gamma$ in $\mathcal C$. Such homotopy classes are elements of the first fundamental group $\pi_1$ of $\R^3\setminus H$.

Given a subset $\mathcal C$ of $\mathcal C_H$ closed by homotopy, a relatively closed subset $K$ of $\R^3\setminus H$ is a \mbox{$\mathcal C$-spanning} set of $H$ if
$K\cap\gamma\neq\emptyset$ for every $\gamma$ in $\mathcal C$.
\end{defn}

\begin{figure}
\begin{center}
\includegraphics[width=0.9\textwidth]{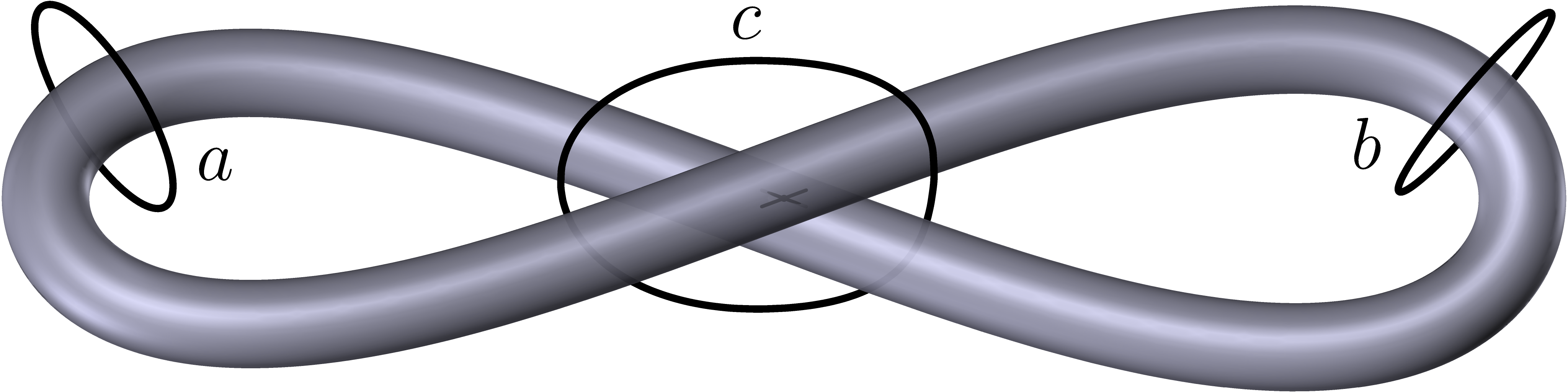}
\end{center}
\caption{An appropriate choice of homotopy classes determines which holes of a bounding loop with points of self-contact are covered. For the particular loop depicted here, which is subject to self contact without interpenetration at the central crossing point (black cross), if we seek a spanning set relative to the homotopy class of the loop $a$ or $b$, spanning surfaces that cover only the hole on the left or on the right will be allowed, respectively. If, instead, we consider the homotopy class of the loop $c$, both holes must be covered by the spanning set.}
\label{fig:spanning}
\end{figure}

The previous definition offers several advantages in comparison to more classical definitions. Most importantly, the set that is spanned by the surface representing the liquid film need not be a one-dimensional structure, it can be any closed set in the ambient space. This allows for the possibility that the surface spans the bounding loop $\loo$ of finite cross-sectional thickness. Another useful feature of this definition is that it allows for choice regarding the number of ``holes to be covered'' by the spanning set. Namely, when the set $H$ has a somewhat complex topology, it is permissible to restrict attention to surfaces that span only significant subregions of $H$. To give a trivial example, illustrated in Figure \ref{fig:spanning}, where $H$ is a loop subject to self contact, it is possible to seek surfaces that span one, the other, or both holes. {Notice that, as shown below, the non-interpenetration constraint still allows for points on the surface of the bounding loop to come into contact, since it entails only a superposition of points on the surface of the rod.}

In particular, the subset $\mathcal D_\loo$ of $\mathcal C_\loo$ containing all $\gamma$ that are not homotopic to a constant ($[\gamma]\neq 1_{\pi_1}$) is closed by homotopy. We will then seek a surface that is a \mbox{$\mathcal D_\loo$-spanning} set of the bounding loop $\loo$. This is a maximal choice in the sense that we cannot include paths homotopic to a constant in the spanning condition. Indeed, since $\loo$ is a compact set, it is easy to see that any \mbox{$\{1_{\pi_1}\}$-spanning} set of $\loo$ required to intersect all constant paths would fill up all of $\R^3\setminus\loo$. This would certainly not represent the behavior of a liquid film.

\subsection{The Kirchhoff--Plateau problem}

The basic step in connecting our mathematical model to the experimental observations consists in finding a minimizer $\tpl$ for the functional
\begin{equation}
E_\mathrm{KP}(\tpl):=E_\mathrm{loop}(\tpl)+\inf\big\{E_\mathrm{film}(S):S\text{ is a \mbox{$\mathcal D_\loo$-spanning} set of }\loo\big\}\,,
\end{equation}
where $\tpl$ belongs to a suitable subset $U$ of $V$, that encodes the additional physical and topological constraints on the rod modeling the bounding loop. These constraints are: (i) the closure of the midline; (ii) {the global gluing conditions;} (iii) global non-interpenetration of matter; (iv) the knot type of the midline.

When we combine the closure constraint with the definition of the Kirchhoff--Plateau functional $E_\mathrm{KP}$, a strong competition between the action of the spanning film and the elastic response of the bounding loop can arise. In a typical situation, the curvature of the loop tends to be minimized, producing somewhat wider configurations which, in turn, require spanning films with larger surface areas. The equilibrium shape is strongly influenced by the relative strength of surface tension with respect to the properties of the filament, but the proper inclusion of the global non-interpenetration constraint guarantees the physical relevance of the solutions the existence of which we establish.

We will now present, following Schuricht~\cite{Sch02}, the precise formulation of the aforementioned constraints and a lemma, proving that the set $U$ is weakly closed in $V$, as this is the essential property needed to establish the tractability of those constraints within our variational approach.

The fact that the midline is a closed curve can be readily expressed by
\begin{equation}\label{eq:closure}
\vc x[\tpl](L)=\vc x[\tpl](0)=\vc x_0\,,
\end{equation}
which we supplement with
\begin{equation}\label{eq:clamping}
\vc t[\tpl](L)=\vc t[\tpl](0)=\vc t_0\qquad\text{and}\qquad\vc d[\tpl](0)=\vc d_0\,.
\end{equation}
These clamping conditions \eqref{eq:closure}--\eqref{eq:clamping} are important in view of the preferred direction associated with the gravitational acceleration $\vc g$. Indeed, the weight term breaks the invariance of our problem under rigid rotations. Moreover, \eqref{eq:closure}--\eqref{eq:clamping} effectively describe the physical operation of holding the flexible structure at one point with tweezers. The right-hand sides of \eqref{eq:closure}--\eqref{eq:clamping}, namely $\vc x_0$, $\vc t_0$, and $\vc d_0$, are referred to as the clamping parameters.

Note that we do not require that $\vc d[\tpl](L)=\vc d[\tpl](0)$, since the rod can be glued fixing an arbitrary angle between $\vc d[\tpl](L)$ and $\vc d[\tpl](0)$ (local gluing condition) while simultaneously respecting the clamping conditions. To prescribe how many times the ends of the rod are twisted before being glued together, 
{we define the integer {link type} of the closed rod as the linking number of the closed midline $\vc x[\tpl]$ and the curve $\vc x[\tpl]+\epsilon\vc d[\tpl]$. Although a sufficiently small $\epsilon$ can always be chosen, the curve may need to be closed as indicated by Schuricht~\cite[Sect.~4.1]{Sch02}. The global gluing conditions are then fixed by prescribing the angle between $\vc d[\tpl](L)$ and $\vc d[\tpl](0)$ and the link type of the closed rod.}

The non-interpenetration of matter can be enforced through the global injectivity condition
\[
\int_\Omega\det\frac{\de\vc p[\tpl](s,\zeta_1,\zeta_2)}{\de(s,\zeta_1,\zeta_2)}\,d(s,\zeta_1,\zeta_2)\leq \ha^3(\vc p[\tpl](\Omega))\,,
\]
which, by \eqref{eq:configuration}, is equivalent to
\begin{equation}\label{eq:glob-inj}
\int_\Omega(1-\zeta_1\kappa_2(s)+\zeta_2\kappa_1(s))\,d(s,\zeta_1,\zeta_2)\leq \ha^3(\vc p[\tpl](\Omega))\,.
\end{equation}
Note that although condition \eqref{eq:glob-inj} implies the global injectivity of $\vc p[\tpl]$ only on $\mathrm{int}(\Omega)$, it is sufficient for our purposes, since we wish to allow for self-contact of the bounding loop, preventing only interpenetration. Clearly, global injectivity implies local injectivity, but the converse is not true. The local property is only a necessary condition for the global injectivity to hold. Indeed, as depicted in Figure \ref{fig:injectivity}, global injectivity may fail due to the overlapping of cross-sections belonging to regions of the rod that lie far apart along the midline and in which local injectivity holds true.

To encode the knot type of the midline, we invoke the notion of isotopy class for closed curves.
\begin{defn}
Let $\vc x_i:[0,L]\to\R^3$, $i=1,2$, be two continuous curves with $\vc x_i(L)=\vc x_i(0)$.
The curves $\vc x_1$ and $\vc x_2$ are called isotopic (denoted $\vc x_1\simeq\vc x_2$) if there are open neighborhoods $N_1$ of $\vc x_1([0,L])$ and $N_2$ of $\vc x_2([0,L])$ and a continuous mapping $\Phi:N_1\times[0,1]\to\R^3$ such that $\Phi(N_1,\tau)$ is homeomorphic to $N_1$ for all $\tau$ in $[0,1]$ and
\[
\Phi(\cdot,0)=\mathrm{Identity}\,,\quad\Phi(N_1,1)=N_2\,,\quad\text{and}\quad
\Phi(\vc x_1([0,L]),1)=\vc x_2([0,L])\,.
\]
\end{defn}
By means of isotopy classes we can encode the knot type of the bounding loop as follows. We fix a continuous mapping $\vc\ell:[0,L]\to\R^3$ such that $\vc\ell(L)=\vc\ell(0)$ and we say that an element $\tpl$ of $V$, for which the closure and clamping conditions \eqref{eq:closure}--\eqref{eq:clamping} hold, has the knot type of $\vc\ell$ if
\begin{equation}\label{eq:knot}
\vc x[\tpl]\simeq\vc\ell\,.
\end{equation}

We can now prove a theorem about the non-interpenetration constraint and our basic lemma about the weak closure of the set $U$ of competitors for our minimization problem. 
\begin{thm}\label{thm:injectivity}
Let $\tpl=((\kappa_1,\kappa_2,\omega),\vc x_0,\vc t_0,\vc d_0)$ in $V$ be such that $E_\mathrm{loop}(\tpl)<+\infty$. Then the configuration map $\vc p[\tpl]$ is locally injective on $\mathrm{int}(\Omega)$. 
Furthermore, this mapping is open on $\mathrm{int}(\Omega)$. 
If, in addition, $\tpl$ satisfies \eqref{eq:glob-inj}, then $\vc p[\tpl]$ is also globally injective on $\mathrm{int}(\Omega)$. 
\end{thm}
\begin{proof}
Let $\tpl$ in $V$ be fixed and consider a point $(\bar{s},\bar{\zeta}_1,\bar{\zeta}_2)$ in $\mathrm{int}(\Omega)$.
Having defined
\[
{A}_{\varepsilon,\delta}:=\big\{(s,\zeta_1,\zeta_2)\in\Omega : |s-\bar{s}|<\varepsilon,\;|\zeta_1-\bar{\zeta}_1|<\delta,\;|\zeta_2-\bar{\zeta}_2|<\delta\big\}\,,
\]
there exist constants $\varepsilon>0$ and $\delta>0$ such that
\begin{equation}\label{eq:well-inside}
\overline{A}_{\varepsilon,3\delta}\subset\mathrm{int}(\Omega)\,.
\end{equation}

By our assumptions on the material cross-sections, there exists $R>0$ such that $|\zeta_1|<R$ and $|\zeta_2|<R$ for any $(s,\zeta_1,\zeta_2)$ in $\Omega$. By condition \eqref{eq:well-inside}, $\delta<R$ and it is possible to choose $\varepsilon$ such that
\begin{equation}\label{eq:smallness}
\int_{\bar s-\varepsilon}^{\bar s+\varepsilon}|\omega(s)|\,ds<\frac{\delta}{3R}
\qquad
\text{and}\qquad 
\vc t[\tpl](s_1)\cdot\vc t[\tpl](s_2)>\frac{1}{2}
\end{equation}
for every $s_1$ and $s_2$ in $(\bar s-\varepsilon,\bar s+\varepsilon)$.
Assuming $E_\mathrm{loop}(\tpl)<+\infty$ and recalling condition \eqref{eq:ni} together with \eqref{eq:well-inside}, we infer that
\begin{equation*}
1+\zeta_2\kappa_1(s)-\zeta_1\kappa_2(s)\geq 2\delta(|\kappa_1(s)|+|\kappa_2(s)|)
\end{equation*}
for almost every $s$ in $(\bar s-\varepsilon,\bar s+\varepsilon)$ and all corresponding $(s,\zeta_1,\zeta_2)$ in ${A}_{\varepsilon,\delta}$.

We now show that $\vc p[\tpl]$ is injective on $A_{\varepsilon,\delta}$. Given two points $(a,\zeta^a_1,\zeta^a_2)$ and $(b,\zeta^b_1,\zeta^b_2)$ in ${A}_{\varepsilon,\delta}$, we assume that
\begin{equation}\label{eq:injectivity}
\vc p[\tpl](a,\zeta^a_1,\zeta^a_2)=\vc p[\tpl](b,\zeta^b_1,\zeta^b_2)\,.
\end{equation}
If $a=b$, we use the definition \eqref{eq:configuration} of the configuration mapping to find that $\zeta^a_1=\zeta^b_1$ and $\zeta^a_2=\zeta^b_2$. We then argue by contradiction, by assuming $a<b$, and we set
\[
\Delta s:=b-a\qquad\text{and}\qquad\tilde{\vc p}(s):=\vc p[\tpl](s,\zeta_1^b,\zeta_2^b)\,.
\]
Then, using \eqref{eq:Cauchy} and \eqref{eq:configuration}, we obtain
\begin{align}
\tilde{\vc p}(b)-\tilde{\vc p}(a)&\mbox{}=\Delta s\int_0^1\tilde{\vc p}'(a+t\Delta s)\,dt\notag\\
&\mbox{}=\Delta s\int_0^1\big[(1-\zeta_1^b\kappa_2+\zeta_2^b\kappa_1)\vc t[\tpl]-\omega\zeta_2^b\vc d[\tpl]+\omega\zeta_1^b(\vc t[\tpl]\times\vc d[\tpl])\big]\,dt\,,\label{eq:deltap}
\end{align}
where $a+t\Delta s$ is the argument of all functions in the second line.
Since 
\[
\big(\vc p[\tpl](b,\zeta^b_1,\zeta^b_2)-\vc p[\tpl](a,\zeta^a_1,\zeta^a_2)\big)\cdot\vc t[\tpl](a)=\big(\tilde{\vc p}(b)-\tilde{\vc p}(a)\big)\cdot\vc t[\tpl](a)\,,
\]
to obtain a contradiction with \eqref{eq:injectivity} it suffices to show that
\begin{equation}\label{eq:contradiction}
\big(\tilde{\vc p}(b)-\tilde{\vc p}(a)\big)\cdot\vc t[\tpl](a)>0\,,
\end{equation}
thereby proving that $a=b$ and that $\vc p[\tpl]$ is locally injective.

On introducing
\[
\alpha_1=\int_a^b|\kappa_1(s)|\,ds\,,\qquad\alpha_2=\int_a^b|\kappa_2(s)|\,ds\,,\qquad\text{and}\qquad\alpha_3=\int_a^b|\omega(s)|\,ds\,,
\]
it follows from \eqref{eq:Cauchy} that
\begin{equation}\label{eq:est1}
\sup_{s\in(a,b)}|\vc d[\tpl](s)\cdot\vc t[\tpl](a)|\leq\frac{\alpha_2+\alpha_3\alpha_1}{1-\alpha_3^2}
\end{equation}
and, analogously, that
\begin{equation}\label{eq:est2}
\sup_{s\in(a,b)}|\vc t[\tpl](s)\times\vc d[\tpl](s)\cdot\vc t[\tpl](a)|\leq\frac{\alpha_1+\alpha_3\alpha_2}{1-\alpha_3^2}\,.
\end{equation}

On multiplying \eqref{eq:deltap} by $\vc t[\tpl](a)$ and, recalling \eqref{eq:smallness}, \eqref{eq:est1}, and \eqref{eq:est2}, we easily obtain
\begin{equation}\label{eq:est3}
\big(\tilde{\vc p}(b)-\tilde{\vc p}(a)\big)\cdot\vc t[\tpl](a)
\geq \Delta s \int_0^1\bigg(\delta(|\kappa_1|+|\kappa_2|)-R|\omega|\frac{\alpha_1+\alpha_2}{1-\alpha_3}\bigg)\,dt\geq\frac{\delta}{2}(\alpha_1+\alpha_2)\,.
\end{equation}
Now, if $\alpha_1+\alpha_2>0$, \eqref{eq:est3} implies \eqref{eq:contradiction}. Otherwise, $\kappa_1(s)$ and $\kappa_2(s)$ must both vanish for almost every $s$ in $(a,b)$. This means that $\vc t[\tpl](s)=\vc t[\tpl](a)$ for every $s$ in $(a,b)$; hence,
\[
\big(\tilde{\vc p}(b)-\tilde{\vc p}(a)\big)\cdot\vc t[\tpl](a)=\Delta s>0\,,
\]
which is \eqref{eq:contradiction}.

We have therefore established the local injectivity of $\vc p[\tpl]$ and moreover that, being continuous, $\vc p[\tpl]$ is an open mapping on $\mathrm{int}(\Omega)$.  Given the local injectivity, the global injectivity follows from condition \eqref{eq:glob-inj}, as proved by Schuricht \cite[Theorem 3.8]{Sch02}.
\end{proof}

\begin{lem}\label{lem:weaklyclosed}
Let a continuous mapping $\vc\ell:[0,L]\to\R^3$ with $\vc\ell(L)=\vc\ell(0)$, the {global gluing conditions}, and a real constant $M$ be given. Let also the clamping parameters $\vc x_0$, $\vc t_0$, and $\vc d_0$ belonging to $\R^3$ be given. Then the set
\[
U:=\big\{
\tpl=((\kappa_1,\kappa_2,\omega),\vc x_0,\vc t_0,\vc d_0)\in V:E_\mathrm{loop}(\tpl)<M\text{ and \eqref{eq:closure}--\eqref{eq:knot} hold}
\big\}
\]
is weakly closed in $V$.
\end{lem}
\begin{proof}
If $U=\emptyset$ the assertion is true. If $U\neq\emptyset$, then the proof is a mere restatement of those of Lemma 3.9, Lemma 4.5, and Lemma 4.6 of Schuricht \cite{Sch02}.
\end{proof}

\begin{figure}
\begin{center}
\includegraphics[width=0.95\textwidth]{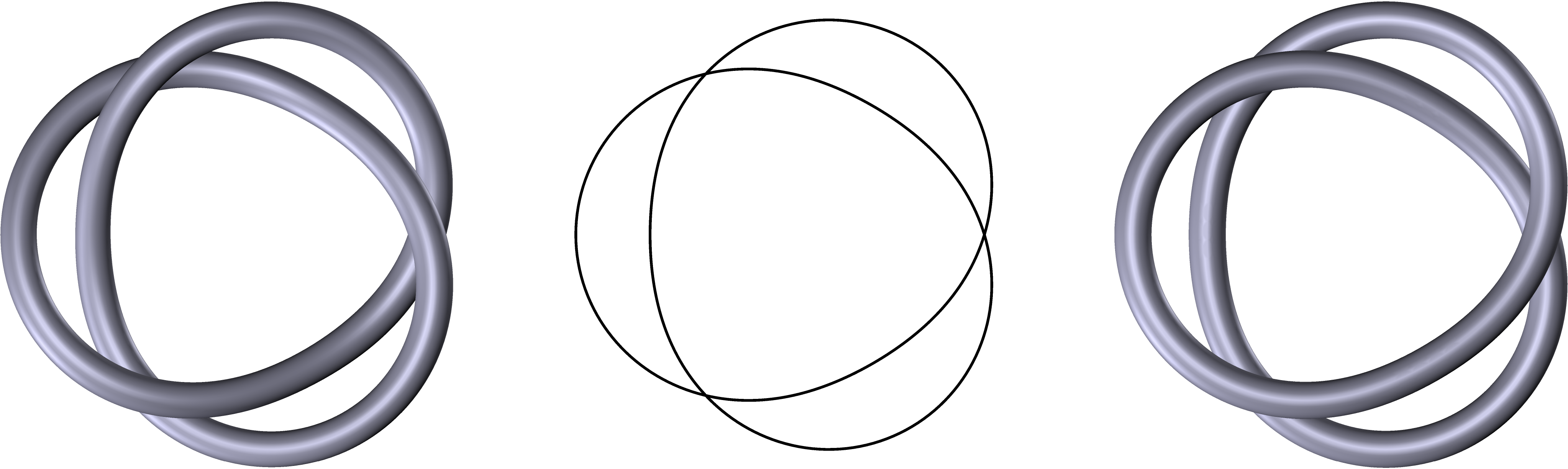}
\end{center}
\caption{A nonvanishing cross-sectional thickness is essential for distinguishing knot types in the presence of self-contact. As a trivial example, the trefoil knot (left) and the unknot (right) are clearly distinct when the cross-sections have a nonvanishing thickness, even in the presence of self-contact. On the contrary, in the vanishing-thickness limit (center), reaching self-contact from each of the two knotted configurations produces topologically equivalent structures and the distinction between a knot and an unknot is lost.}
\label{fig:thickness}
\end{figure}

We stress that the global injectivity condition \eqref{eq:glob-inj} is crucial in ensuring the closure proved in Lemma \ref{lem:weaklyclosed}. Indeed, just as the physical non-interpenetration of matter makes it impossible to change the type of a knot without tearing the loop that forms it, global injectivity entails that the knot type is preserved when passing to the limit in a sequence within the constrained set $U$. In particular, the fact that constrained sets defined by different knot types are well separated in $V$ stems from the nonvanishing thickness of the cross-sections and condition \eqref{eq:glob-inj} and would be inevitably lost in the vanishing-thickness limit, as illustrated in Figure \ref{fig:thickness}.

%%%%%%%  SECTION 3  %%%%%%%%%%%%%%%%%%%%%%%%%%%%%%%
\section{Existence {of global minimizers}}\label{sec:existence}

In this section, we prove the existence of a solution to the Kirchhoff--Plateau problem. Specifically, we establish the existence of a global minimizer in $V$ for the total energy $E_\mathrm{KP}$, under appropriate physical and topological conditions. 
As necessary steps towards this goal, we first prove the existence of an energy-minimizing configuration for the bounding loop in the absence of the liquid film and then we demonstrate the existence of an area-minimizing spanning surface for a rigid bounding loop. This constitutes a somewhat more physical version of the Plateau problem, since the bounding loop is treated as a three-dimensional object. 

Those two results are quite straightforward, given the results of the previous section and the general theorem proved by De Lellis, Ghiraldin \& Maggi \cite[Theorem 4]{DeLGhi14}.
On the other hand, the proof of our main existence result requires establishing two lemmas in which the arguments of De Lellis, Ghiraldin \& Maggi \cite{DeLGhi14} are adapted to our situation, wherein the set that is spanned by the surface changes along minimizing sequences for the Kirchhoff--Plateau energy. Based on these results, we are able to verify the lower semicontinuity property that is needed to establish the existence of a solution of the Kirchhoff--Plateau problem. Moreover, Lemma \ref{lem:diagonal} provides an interesting extension of the main compactness result of De Lellis, Ghiraldin \& Maggi \cite{DeLGhi14}, which could be used also in other contexts when studying the convergence of minimal surfaces induced by the convergence of the structure spanned by the surfaces.

\begin{thm}[Minimization of the loop energy]\label{thm:loop}
Let a continuous mapping $\vc\ell:[0,L]\to\R^3$ satisfying $\vc\ell(L)=\vc\ell(0)$, the {global gluing conditions}, and the clamping parameters $\vc x_0$, $\vc t_0$, and $\vc d_0$ belonging to $\R^3$ be given. If there exists $\tilde{\tpl}=(\tilde{\tpl}_1,\vc x_0,\vc t_0,\vc d_0)$ in $V$ such that $E_\mathrm{loop}(\tilde{\tpl})<+\infty$ and that complies with all of the constraints \eqref{eq:closure}--\eqref{eq:knot}, then there exists
a minimizer $\tpl=(\tpl_1,\vc x_0,\vc t_0,\vc d_0)$ belonging to $V$ for the loop energy functional $E_\mathrm{loop}$ and obeying the same constraints.
\end{thm}
\begin{proof}
Consider a minimizing sequence $\{\tpl_k\}$ such that $E_\mathrm{loop}(\tpl_k)<M$ for some $M$ in $\R$, with $\tpl_k$ belonging to the subset $U$ of the reflexive Banach space $V$ introduced in Lemma~\ref{lem:weaklyclosed}. The existence of a competitor $\tilde{\tpl}$ guarantees that $U$ is not empty. By the coercivity of $E_\mathrm{loop}$, $U$ is bounded in $V$. Hence, it is possible to extract a weakly converging subsequence $\tpl_{k_i}\wto\tpl$. Since, by Lemma~\ref{lem:weaklyclosed}, $U$ is weakly closed, $\tpl$ must belong to $U$ and, by the weak lower semicontinuity of  $E_\mathrm{loop}$, we conclude that $\liminf_i E_\mathrm{loop}(\tpl_{k_i})\geq E_\mathrm{loop}(\tpl)$. Since $\{\tpl_{k_i}\}$ is a minimizing sequence, this proves that the weak limit $\tpl$ is indeed a minimizer.
\end{proof}

From now on, for any  $\tpl$ in $V$, we denote by $\loo$ the image of $\Omega$ under the configuration map $\vc p[\tpl]$  and by $\mathcal S[{\tpl}]$ the set of all \mbox{$\mathcal D_{\loo}$-spanning} sets of $\Lambda [{\tpl}]$, with $\mathcal D_\loo$ defined as in Section \ref{sec:spanning}. Moreover, for any $\gamma \colon S^1 \to \mathbb R^3$ and for any $r>0$ we denote by $U_r(\gamma)$ the tubular neighborhood  of radius $r$ around $\gamma$.

\begin{thm}[Minimization of the surface energy]\label{thm:Plateau}
Fix $\tpl$ in $V$. If
\[
\alpha:=\inf\big\{E_\mathrm{film}(S):S\in\mathcal S[{\tpl}]\big\}<+\infty\,,
\]
then there exists a relatively closed subset $M[\tpl]$ of $\R^3\setminus\loo$ that is a \mbox{$\mathcal D_\loo$-spanning} set of $\loo$ with $E_\mathrm{film}(M[\tpl])=\alpha$.
Furthermore, $M[\tpl]$ is an $(\mathbf M,0,\infty)$-minimal set in $\R^3\setminus\loo$ in the sense of Almgren \cite{Alm76}. {In particular, $M[\tpl]$ is countably $\ha^2$-rectifiable.} 
\end{thm}
\begin{proof}
The proof is a direct application of the general theorem by De Lellis, Ghiraldin \& Maggi \cite[Theorem 4]{DeLGhi14}. {The countably $\ha^2$-rectifiability of $M[\tpl]$ follows from the regularity of $(\mathbf M,0,\infty)$-minimal sets established by Almgren \cite{Alm76}.}
\end{proof}

\begin{defn}\label{defn:Hausdorff}
Let $A$ and $B$ be two closed nonempty subsets of a metric space $(M,d_M)$. The Hausdorff distance between $A$ and $B$ is defined as
\[
d_\ha(A,B):=\max\big\{\sup_{a\in A}\inf_{b\in B}d_M(a,b),\sup_{b\in B}\inf_{a\in A}d_M(a,b)\big\}\,.
\]
The topology induced by $d_\ha$  on the space of closed nonempty subsets of $M$ is the Hausdorff topology.
\end{defn}
\begin{lem}\label{lem:diagonal}
Consider a sequence of closed nonempty subsets $\Lambda_k$ of $\R^3$ converging in the Hausdorff topology to the closed set $\Lambda\neq\emptyset$. Assume that, for every $k$ in $\N$, we have countably $\ha^2$-rectifiable sets $S_k$ belonging to $\mathcal P(\Lambda_k)$, where $\mathcal P(\Lambda_k)$ is a good class in the sense of De Lellis, Ghiraldin \& Maggi \cite[Definition 1]{DeLGhi14}, and such that
\[
\ha^2(S_k)=\inf\big\{\ha^2(S):S\in\mathcal P(\Lambda_k)\big\}<+\infty\,.
\]
Then the measures $\mu_k:=\ha^2\res S_k$ constitute a bounded sequence, $\mu_k\stackrel{*}{\rightharpoonup}\mu$ up to the extraction of a subsequence, and the limit measure satisfies
\[
\mu\geq\ha^2\res S_\infty\,,
\] 
where $S_{\infty}:={\rm spt}(\mu)\setminus\Lambda$ is a countably $\ha^2$-rectifiable set.
\end{lem}
\begin{proof}
The proof of this lemma requires minor modifications of the proof of Theorem 2 of De Lellis, Ghiraldin \& Maggi \cite{DeLGhi14}. It is sufficient to observe that the convergence of $\{\Lambda_k\}$ ensures that, whenever $\vc x\in S_\infty$, we have $d(\vc x,\Lambda_k)>0$ for large enough $k$. Then, all arguments in the proof of Theorem 2 of De Lellis, Ghiraldin \& Maggi \cite{DeLGhi14} are recovered asymptotically.
\end{proof}

\begin{lem}\label{lem:intersection}
Let $\{{\tpl}_k\}$ be a sequence weakly converging to ${\tpl}$ in the subset $U$ of $V$ introduced in Lemma~\ref{lem:weaklyclosed}, let $S_k$ be an element of $\mathcal S[{\tpl}_k]$, and fix $\gamma$ in $\mathcal D_{\loo}$. 
Then, for any $\varepsilon>0$ such that $U_{2\varepsilon}(\gamma)$ is contained in $\mathbb R^3\setminus \loo$, there exists $M=M(\varepsilon)>0$ such that, for any $k$ large enough,
\begin{equation}\label{key}
\mathcal H^2(S_k \cap U_\varepsilon(\gamma))\ge M\,.
\end{equation}
\end{lem}
\begin{proof}
Let us fix $\varepsilon>0$ such that $U_{2\varepsilon}(\gamma)$ is contained in $\mathbb R^3\setminus \loo$, denote by $B_\varepsilon(\vc 0_2)$ the open disk of $\R^2$ with radius $\varepsilon$ and centered at the origin of $\R^2$, and consider a diffeomorphism $\Phi\colon S^1 \times B_\varepsilon(\vc 0_2) \to U_{\varepsilon}(\gamma)$ such that $\Phi_{|_{S^1\times \{\vc 0_2\}}}=\gamma$. Let $\vc y$ belong to $B_\varepsilon(\vc 0_2)$ and set $\gamma_{\vc y}:=\Phi_{|_{S^1\times \{\vc y\}}}$. Then $\gamma_{\vc y}$ in $[\gamma]$ represents an element of $\pi_1(\mathbb R^3 \setminus \loo)$. 

Let $\vc x_k$ and $\vc x$ denote the midlines corresponding respectively to ${\tpl}_k$ and ${\tpl}$. Since $\{{\tpl}_k\}$ converges weakly to ${\tpl}$ in $U$, $\{\vc x_k\}$ converges to $\vc x$ strongly in $W^{1,p}([0,L];\mathbb R^3)$. In particular, $\{\vc x_k\}$ converges to $\vc x$ uniformly on $[0,L]$, which implies that, for $k$ sufficiently large, $\Lambda[{\tpl}_k]$ is contained in a neighborhood $W$ of $\Lambda[{\tpl}]$ with $W\cap U_\varepsilon(\gamma)=\emptyset$. Hence, for such $k$ and $\varepsilon$ it follows that, for any $\vc y$ in $B_\varepsilon(\vc 0_2)$, $\gamma_{\vc y}$ belongs to $\mathbb R^3\setminus \Lambda[{\tpl}_k]$, which yields $S_k\cap \gamma_{\vc y}\ne \emptyset$ because $S_k$ is in $\mathcal S[{\tpl}_k]$. 

Take $\vc\pi : S^1 \times B_\varepsilon(\vc 0_2) \to B_\varepsilon(\vc 0_2)$ as the projection on the second factor and let $\hat{\vc\pi}:=\vc\pi \circ \Phi^{-1}$. Then, $\hat{\vc\pi}$ is Lipschitz-continuous and $B_\varepsilon(\vc 0_2)$ is contained in $\hat{\vc \pi}(S_k \cap U_\varepsilon(\gamma))$, which entails that
\[
\pi \varepsilon^2=\ha^2(B_\varepsilon(\vc 0_2))\le\ha^2(\hat{\vc \pi}(S_k \cap U_\varepsilon(\gamma))\le ({\rm Lip}\,\hat{\vc\pi})^2\ha^2(S_k \cap U_\varepsilon(\gamma))\,.
\]
We thus conclude that
\[
\mathcal H^2(S_k \cap U_\varepsilon(\gamma))\ge \frac{\pi \varepsilon^2}{({\rm Lip}\,\hat{\vc\pi})^2}\,,
\]
which establishes the inequality \eqref{key}.
\end{proof}

\begin{thm}[Main existence result]\label{thm:existence}
Let a continuous mapping $\vc\ell:[0,L]\to\R^3$ satisfying $\vc\ell(L)=\vc\ell(0)$, the {global gluing conditions}, and the clamping parameters $\vc x_0$, $\vc t_0$, and $\vc d_0$ belonging to $\R^3$ be given. If there exists $\tilde{\tpl}=(\tilde{\tpl}_1,\vc x_0,\vc t_0,\vc d_0)$ in $V$ such that $E_\mathrm{KP}(\tilde{\tpl})<+\infty$ and which complies with the constraints \eqref{eq:closure}--\eqref{eq:knot}, then there exists a solution $\tpl=(\tpl_1,\vc x_0,\vc t_0,\vc d_0)$ to the Kirchhoff--Plateau problem belonging to $V$, namely a minimizer of the total energy $E_\mathrm{KP}$ satisfying \eqref{eq:closure}--\eqref{eq:knot}.
Furthermore, the spanning surface $M[\tpl]$ associated with the energy minimizing configuration by Theorem~\ref{thm:Plateau} is an $(\mathbf M,0,\infty)$-minimal set in $\R^3\setminus\loo$ in the sense of Almgren \cite{Alm76}. {In particular, $M[\tpl]$ is countably $\ha^2$-rectifiable.}
\end{thm}
\begin{proof}
Consider a minimizing sequence $\{{\tpl}_k\}$ for $E_{\rm KP}$ such that $E_{\rm KP}({\tpl}_k) < M$ for some $M$ in $\mathbb R$. In particular, we then have that $E_{\rm loop}({\tpl}_k) < M$ and we can choose ${\tpl}_k$ in the subset $U$ of $V$ introduced in Lemma \ref{lem:weaklyclosed}. As in the proof of Theorem~\ref{thm:loop}, we can extract a weakly convergent subsequence ${\tpl}_{k_i} \rightharpoonup  \overline {\tpl}$ with $\overline{\tpl} \in U$. To complete the proof, it remains to establish that $E_{\rm KP}$ is weakly lower semicontinuous on $V$. Given the weak lower semicontinuity of $E_{\rm loop}$, this is tantamount to proving that the functional
\begin{equation}\label{funct}
{\tpl} \mapsto \inf\big\{\mathcal H^2(S) : \textrm{$S \in\mathcal S[{\tpl}]$}\big\}
\end{equation} 
is weakly lower semicontinuous.

Fix a weakly convergent sequence ${\tpl}_k \rightharpoonup {\tpl}$ in $U$. Let $S_k$ belonging to $\mathcal S[{\tpl}_k]$ be given by Theorem \ref{thm:Plateau} such that 
\[
\mathcal H^2(S_k) =\inf\big\{\mathcal H^2(S) : \textrm{$S \in\mathcal S[{\tpl}_k]$}\big\}\,.
\]
Without loss of generality, we can assume that $\mathcal H^2(S_k)\le C$ for some $C>0$. For any $k$ in $\N$, let $\mu_k:=\mathcal H^2 \res S_k$. Then, up to the extraction of a subsequence, we have $\mu_k \stackrel{*}{\rightharpoonup}\mu$ on $\mathbb R^3$ and we can set $S_0:={\rm spt}(\mu)\setminus \Lambda[{\tpl}]$. 
On applying Lemma \ref{lem:diagonal} with $\Lambda_k=\Lambda[\tpl_k]$, we deduce that 
\begin{equation*}%\label{step1}
\mu \ge \mathcal H^2\res S_0 \textrm{ on subsets of }\mathbb R^3 \setminus \Lambda [{\tpl}]\,.
\end{equation*} 
We next show that $S_0$ belongs to $\mathcal S[{\tpl}]$. Assume by contradiction that there exists $\gamma$ in $\mathcal D_{\Lambda[{\tpl}]}$ with $\gamma \cap S_0= \emptyset$ and pick $\varepsilon$ as given by Lemma \ref{lem:intersection}.
We then find that $\mu(U_{\varepsilon}(\gamma))=0$ and, therefore, that
\[
\lim_k\mathcal H^2(S_k \cap U_{\varepsilon}(\gamma))=0\,,
\]
which contradicts the thesis \eqref{key} of Lemma \ref{lem:intersection}. 
Hence, we obtain the chain of inequalities 
\begin{align*}
\liminf_k\inf\{\mathcal H^2(S) : \textrm{$S \in\mathcal S[{\tpl}_k]$}\} &\mbox{}\ge \liminf_k\mathcal H^2(S_k)\\
&\mbox{}=\liminf_k\mathcal \mu_k(\mathbb R^3)\ge \mu(\mathbb R^3)
\ge \mathcal H^2(S_0)\ge \inf\{\mathcal H^2(S) : \textrm{$S \in\mathcal S[{\tpl}]$}\}\,,
\end{align*}
which establishes the lower semicontinuity of the functional \eqref{funct}.
\end{proof}

%%%%%%%  SECTION 4  %%%%%%%%%%%%%%%%%%%%%%%%%%%%%%%
\section{Concluding remarks}\label{sec:conclusion}

We introduced a mathematical model for experiments in which a thin filament in the form of a closed loop is spanned by a liquid film. In a variational setting, the model is defined by the sum of the energies of the different components of this system. While the loop is modeled as a nonlinearly elastic rod which is inextensible and unshearable, namely a Kirchhoff rod, the elasticity of the liquid film is described by a homogeneous surface tension. 

Following Schuricht \cite{Sch02}, we required that the loop satisfy the physical constraints of non-interpe\-netration of matter, fixed {global gluing conditions}, and fixed knot type. A crucial point in this treatment is that the cross-sectional thickness of the loop is nonzero, implying that it occupies a nonvanishing volume. This led us to consider a somewhat more physical version of the Plateau problem where the liquid film is represented by a surface with a boundary that is not prescribed, but is free to move along the lateral surface of the three-dimensional bounding loop. Our choice to retain a nonvanishing cross-sectional thickness of the loop, while attributing a vanishing thickness to the liquid film, is justified by the typically large separation between those two length-scales.

By exploiting the framework recently proposed by De Lellis, Ghiraldin \& Maggi \cite{DeLGhi14} for the Plateau problem,
we established the existence of a global minimizer, namely a stable solution, for the coupled Kirchhoff--Plateau problem, in which the boundary of the liquid film lies on the lateral surface of the deformable bounding loop.  This was achieved by means of a dimensional reduction, performed in expressing the total energy of the system as a functional of the geometric descriptors of the rod only. To this end, a strongly nonlocal term, entailing the minimization of the surface energy for a fixed shape of the loop, was introduced and the effectiveness of this strategy guaranteed by the auxiliary proof of existence of a surface realizing such a minimization.
A key step towards the final result is the adaptation, presented in Lemma~\ref{lem:diagonal}, of the main compactness argument of De Lellis, Ghiraldin \& Maggi~\cite{DeLGhi14} needed to deal with the deformability of the bounding loop.

Combining the approaches of Schuricht \cite{Sch02}, for the bounding loop, and of De Lellis, Ghiraldin \& Maggi \cite{DeLGhi14}, concerning the liquid film, we established the existence of a solution that complies with important physical constraints on the topology of the bounding loop and of the surface. Indeed, the latter enjoys the soap film regularity identified by Almgren \cite{Alm76} and Taylor \cite{Tay76}. An important outcome of our analysis is that the existence of a physically relevant solution is obtained irrespective of the relative strength of surface tension compared to the elastic response of the filament from which the bounding loop is made, which instead influences the multiplicity and the qualitative properties of solutions, as discussed below.

Based on the framework introduced above, further investigations can be addressed in two main directions. First, it would be interesting to study the existence of unstable solutions for the Kirchhoff--Plateau problem. We expect that classical techniques used to establish similar results in the context of the Plateau problem will not be easily applicable here, while suitably adapted tools from nonsmooth critical point theory may prove to be required and effective. It would be also important to investigate the unstable solutions generated by bifurcation phenomena, which are expected to be a feature of our model. Indeed, the competition between the action of the spanning film and the elastic response of the filament can trigger the transition between different stability regimes, meanwhile affecting the multiplicity of solutions to the Kirchhoff--Plateau problem.

Secondly, studying a dissipative dynamics of the system considered above represents a challenging task both from the analytical point of view, again due to the lack of smoothness inherent to our setting, and from the mechanical point of view, since a physically consistent representation of the dissipative phenomena at play might not be straightforward. Nevertheless, a dynamical strategy of the proposed kind would likely provide an excellent basis for the implementation of numerical schemes aimed at finding local minima of the Kirchhoff--Plateau functional.
\\
\par
\textbf{Acknowledgments.}
{E.\ F. and G.\ G.\ G. gratefully acknowledge support from the Okinawa Institute of Science and Technology Graduate University with subsidy funding from the Cabinet Office, Government of Japan.
L.\ L. is grateful for the kind hospitality of the Okinawa Institute of Science and Technology Graduate University during the inception of the present work.
The authors thank Jenny Harrison and the anonymous reviewers for their insightful comments.
The authors would also like to thank Stoffel Janssens and David Vazquez Cortes for their help with the preparation of Figure 2.}

%%%%%%%%%%%%%%%%%
%===============Bibliography
%%%%%%%%%%%%%%%%%

%\bibliography{articles-KP}
%\bibliographystyle{abbrv}

\end{document}